\documentclass[11pt]{article}

\usepackage[a4paper,margin=2.5cm]{geometry}
\usepackage{amsmath,amssymb,amsthm}
\usepackage{booktabs,array,tabularx}
\usepackage{graphicx}
\graphicspath{{figures/}}
\usepackage{algorithm,algpseudocode}
\usepackage{float}
\usepackage[hidelinks]{hyperref}
\usepackage{url}
\hypersetup{
  unicode=true,
  pdftitle={NativeScope: Relation-Localized Retrieval over Native Topology with a Correct Anchor},
  pdfauthor={Long Wang}
}

\newtheorem{proposition}{Proposition}

\newcommand{\FQ}{NS-FullQ}
\newcommand{\metric}{\mathrm{Recall}@1024}
\title{NativeScope: Relation-Localized Retrieval over Native Topology with a Correct Anchor}
\author{Long Wang\\
  {\small\href{mailto:alueas@126.com}{\texttt{alueas@126.com}}}\\
  {\small\href{https://orcid.org/0009-0006-8371-1562}{ORCID: 0009-0006-8371-1562}}}
\date{}

\begin{document}
\pagestyle{plain}
\maketitle

\begin{abstract}
Dense retrieval usually ranks text chunks by their semantic similarity to a question. This ignores structure that many data systems already store, including section membership, session boundaries, and native order. We propose NativeScope, a scope-then-rank method for queries with a known anchor and relation. It represents a query as $q\rightarrow(A,r,B)$: the anchor $A$ and relation $r$ select native units through belonging, before, or after operators, and the target term $B$ ranks only chunks that overlap the selected scope. An internal variant, NS-FullQ, ranks the same candidates with the full question. We evaluate both methods on 200 controlled document and memory records derived from QASPER and LongMemEval under a 1,024-token budget. NativeScope attains native-unit recall of 89.28\% for documents and 72.50\% for memories, improving over instance-wide Dense RAG by 42.75 and 22.00 percentage points. NS-FullQ reaches 87.78\% and 68.50\%; its differences from NativeScope are inconclusive, locating the primary gain in relational scoping rather than the shorter ranking query. With automatic Top-1 anchors, memory recall falls to 35.50\%. NativeScope is therefore effective when anchor coordinates and native relations are reliable, but hard scoping inherits errors from the localization interface.

\medskip
\noindent\textbf{Keywords:} relational retrieval; native topology; retrieval-augmented generation; long-term memory; conditional evaluation
\end{abstract}

\section{Introduction}

Retrieval-augmented generation typically finds evidence by comparing a question with text chunks. Some questions, however, constrain location as well as content: the answer may lie in the section containing a paragraph or before a particular dialogue turn. Documents and memory systems often preserve these relations as section identifiers, session boundaries, and native order. Treating the collection as flat text discards information that can directly restrict the search space.

We ask a conditional question: \emph{how much does relational scoping help once the query anchor is known and the stored topology is trustworthy?} The anchor may come from a user selection, a stable record identifier, or a validated upstream locator. Our main evaluation supplies the correct coordinate to isolate retrieval; a separate automatic-localization experiment measures how the method degrades when that condition fails.

NativeScope writes a relational query as $q\rightarrow(A,r,B)$. $A$ identifies a location, $r$ specifies a scope operation over native structure, and $B$ identifies what should be retrieved within that scope. The key step is to use the relation as a candidate-domain constraint before semantic ranking. The full question can also rank candidates within the same scope; we call this internal variant \FQ{}. Relational scoping and the choice between $q$ and $B$ as a ranking input are therefore evaluated separately.

NativeScope connects native coordinates, relation operators, and semantic ranking in a single retrieval interface. Unlike systems that infer a new hierarchy or graph, it executes relations already present in the source data. We implement the same pipeline over document paragraphs and dialogue turns, with shared definitions of coordinates, token budget, and evidence coverage. The evaluation separates the value of relational scoping from the choice of ranking query and from the reliability of anchor localization.

\section{Related Work}

\subsection{Dense retrieval and retrieval-augmented generation}

Dense retrieval and RAG rank semantically related passages and then expose them to a generator \cite{karpukhin2020dpr,lewis2020rag,izacard2021fid}. Iterative systems additionally decide when to retrieve again or how retrieval should interact with reasoning and generation \cite{trivedi2023ircot,jiang2023flare,asai2024selfrag}. NativeScope addresses a different decision inside each retrieval step: which candidates are eligible when the data source already supplies a relation to an anchor.

Long context does not eliminate candidate selection. Models may use relevant information poorly when it appears in the middle of a context \cite{liu2024lost}, and LongRAG organizes long-document retrieval from both global-information and local-fact perspectives \cite{zhao2024longrag}. Our evaluation keeps the encoder and token budget fixed and isolates the effect of candidate scope rather than generation quality.

\subsection{Hierarchy, graph indexes, and retrieval granularity}

Structured retrievers organize candidates through document hierarchies, recursive summaries, entity graphs, or discourse trees \cite{liu2021dhr,sarthi2024raptor,edge2024graphrag,lu2026hichunk,chen2026beyond}. These systems show that flat similarity search can benefit from structured candidate organization, but generally construct or infer an index during preprocessing.

Retrieval quality also depends on the indexed unit and its context: prior work compares document, passage, sentence, and proposition indexes, preserves global context during chunk encoding, or expands hits by semantic continuity \cite{chen2024densex,gunther2024late,langlois2026scar}. NativeScope is orthogonal to these choices. It executes section, session, and order relations already stored by the source, then maps the resulting native units to whatever chunks the index provides.

\subsection{Conversational retrieval, anchors, and long-term memory}

Conversational retrieval must recover an information need from multi-turn context, often through learned conversation representations, query rewriting, or salient entities \cite{mao2024chatretriever,roy2024selfmultirag,shavarani2025entity}. This work suggests a useful separation between locating a relevant context and ranking content within it. NativeScope makes that separation explicit through anchor, relation, and target inputs.

Long-term memory systems add storage policies, hierarchy, temporal organization, and contextual reconstruction to retrieval \cite{packer2023memgpt,zhong2024memorybank,ye2026hmem,tang2026mnemis,wang2026emem}. CueMem is especially close to anchor-centered retrieval because it locates a source turn from a cue and restores related context \cite{wang2026cuemem}. Our experiments hold the decomposition, topology, and encoder fixed to isolate the effect of executing an explicit relation over stored coordinates.

\section{Problem Formulation and Method}
\label{sec:method}

Figure~\ref{fig:pipeline} shows the complete NativeScope pipeline. A relation is executed over native units before eligible units are mapped to text chunks. Semantic similarity ranks only within the resulting relation-localized scope, separating \emph{where to search} from \emph{what to search for}.

\begin{figure}[H]
\centering
\includegraphics[width=\linewidth]{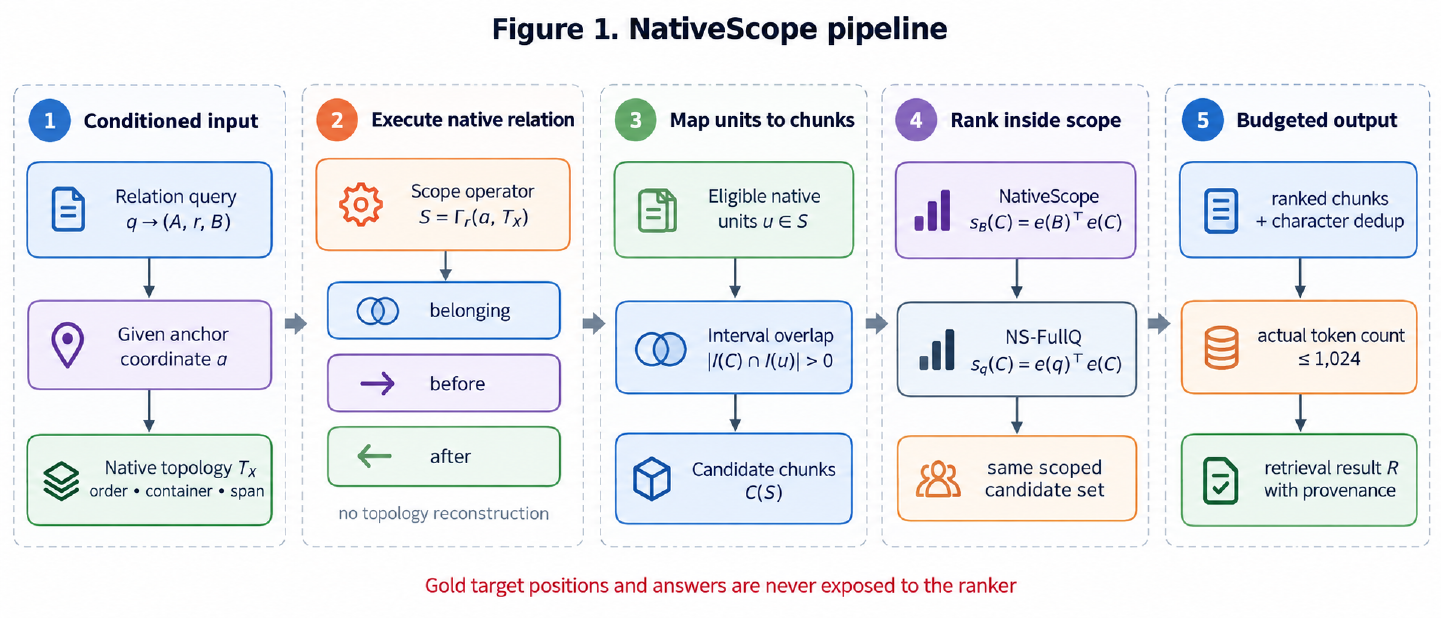}
\caption{NativeScope pipeline. A decomposed query and correct anchor drive a native relation operator. Eligible units are mapped to candidate text chunks by character-interval overlap. NativeScope and NS-FullQ share the same scoping stage and differ only in within-scope ranking input. Gold target locations and answers are not exposed to the ranker.}
\label{fig:pipeline}
\end{figure}

\subsection{Input conditions and native topology}

Let a runtime instance be $X=(U,T_X)$, where $U=\{u_1,\ldots,u_n\}$ contains native units and $T_X$ stores their order, container identifiers, and character intervals. Documents use paragraphs as units and sections as containers; memories use dialogue turns as units and sessions as containers. Let $o(u)$ be unit order, $c(u)$ its container, and $I(u)$ its character interval. The text-chunk index $\mathcal C$ is produced from the original text, and every chunk retains its interval $I(C)$ and a mapping to overlapping native units.

The main task receives query $q$, a frozen decomposition $(A,r,B)$, correct anchor coordinate $a$, and reliable topology $T_X$. Gold target locations and answers are not provided to the ranker. The anchor is additional structural input and is held constant across all anchor-dependent comparisons. Automatic parsing and localization can be connected upstream, but their reliability is evaluated separately from the main retrieval condition.

\subsection{Why observed topology matters}

Pretrained representations may encode common structural semantics, but they can use an instance-specific relation only when the retrieval interface exposes it. Consider two instances with the same observable query and text but different section or session assignments. A retriever that sees only the shared text must return the same candidates for both, even when the relation points to different target units. We refer to this mismatch between represented content and unavailable runtime structure as a local topology-availability gap.

This observation motivates the interface rather than a universal impossibility claim. A sufficiently large budget could return the union of all plausible targets, and some questions remain answerable from content alone. NativeScope instead asks whether explicitly executing an available relation can improve retrieval under a fixed budget.

\subsection{Relation-localized scoping and target ranking}

We implement three operators. The before and after relations use native unit order rather than general temporal reasoning:
\begin{align}
\Gamma_{\mathrm{belonging}}(a,T_X)&=\{u:c(u)=c(a)\},\\
\Gamma_{\mathrm{before}}(a,T_X)&=\{u:1\le o(a)-o(u)\le8\},\\
\Gamma_{\mathrm{after}}(a,T_X)&=\{u:1\le o(u)-o(a)\le8\}.
\end{align}
Directional scopes follow global native order and are not truncated at the current container. Let $S=\Gamma_r(a,T_X)$. Eligible chunks are
\begin{equation}
\mathcal C(S)=\{C\in\mathcal C:\exists u\in S,\ |I(C)\cap I(u)|>0\}.
\end{equation}
This overlap mapping does not trim chunks at a scope boundary, so an eligible boundary chunk may contain a small amount of out-of-scope text. Scoping constrains candidate eligibility rather than enforcing character-level clipping.

With normalized encoder $e$, NativeScope ranks eligible chunks by
\begin{equation}
s_B(C)=e(B)^\top e(C),\qquad C\in\mathcal C(S).
\end{equation}
\FQ{} replaces $B$ with the full question $q$ while keeping the anchor, relation operator, and candidate set identical. The two variants isolate the effect of the ranking input within a common relation-localized framework.

Figure~\ref{fig:operators} illustrates the operators on a common native-unit sequence. Relational selection precedes chunk scoring: only chunks whose character intervals intersect selected units become candidates, regardless of the semantic scores of out-of-scope chunks.

\begin{figure}[H]
\centering
\includegraphics[width=\linewidth]{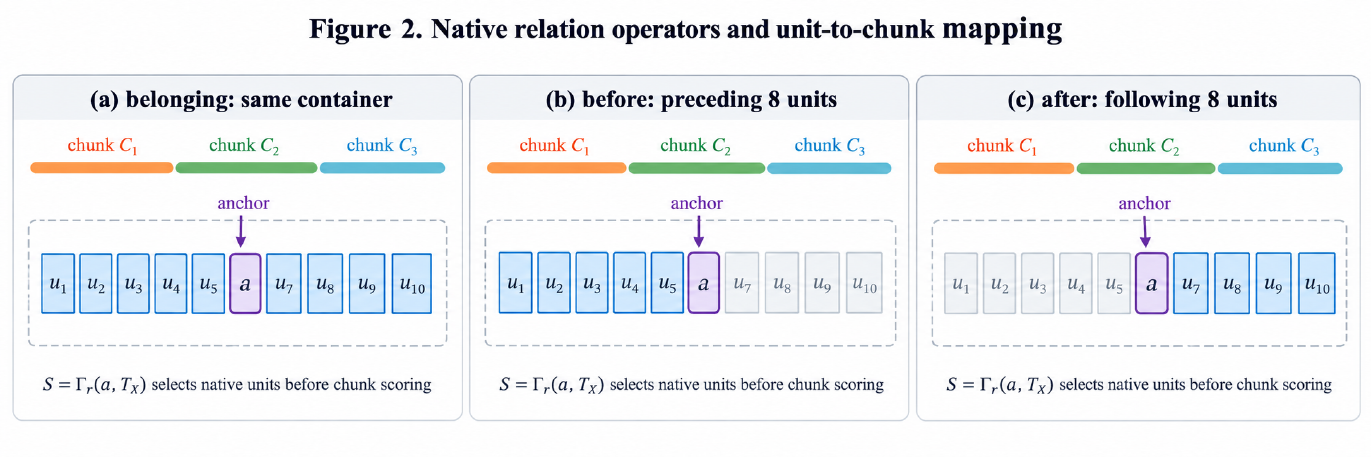}
\caption{Native relation operators and unit-to-chunk mapping. Belonging selects the anchor's container, while before and after select up to eight units in the respective direction under global native order. Colored bars denote chunks that may span several native units.}
\label{fig:operators}
\end{figure}

\subsection{Budget and chunking}

Chunks are ranked by decreasing score, with ties broken by start character, end character, and chunk identifier. Context is assembled in rank order while already returned overlapping characters are deduplicated. The BGE tokenizer counts the actual concatenated text. Retrieval stops when adding the next chunk would exceed 1,024 tokens; it neither skips that chunk to seek a shorter one nor truncates a prefix to fill the remaining budget.

\begin{algorithm}[htbp]
\caption{NativeScope under a correct anchor}
\begin{algorithmic}[1]
\Require Correct anchor $a$, relation $r$, target term $B$, topology $T_X$, chunks $\mathcal C$ with character intervals, budget $L$
\Ensure Ranked returned chunks $R$; budget accounting uses their deduplicated context
\State $S\gets\Gamma_r(a,T_X)$; $R\gets\varnothing$
\State Construct $\mathcal C(S)$ by native character-interval overlap
\State Rank $\mathcal C(S)$ by $s_B$ with deterministic coordinate-based tie breaking
\For{each ranked chunk $C$}
  \State $H\gets\operatorname{DedupeAndJoin}(R\mathbin{+}[C])$
  \If{$\operatorname{TokenCount}(H)>L$}
    \State \textbf{break}
  \EndIf
  \State $R\gets R\mathbin{+}[C]$
\EndFor
\State \Return $R$
\end{algorithmic}
\end{algorithm}

\begin{proposition}[Conditional chunking invariance of the native scope]
If two chunkings leave native units, topology, and the anchor coordinate unchanged, and $\Gamma_r$ does not inspect chunk boundaries, both chunkings yield the same native retrieval scope $S$.
\end{proposition}
\begin{proof}
All inputs to $\Gamma_r$ are unchanged, so its output is unchanged.
\end{proof}
Native-unit ranking would also be invariant if it did not depend on chunking. Our implementation encodes chunks and selects them under a budget, so chunking can change candidates, representations, and complete-evidence coverage even when $S$ is stable. Container-to-unit and unit-to-chunk mappings can be precomputed to avoid re-inferring stored relations; we do not claim a measured latency or memory advantage from this implementation choice.

\section{Data and Experimental Protocol}

\subsection{Controlled data construction and review}

The document domain uses QASPER v0.3 test \cite{dasigi2021qasper}, and the memory domain uses LongMemEval-S cleaned \cite{wu2025longmemeval}; oracle and M splits are not mixed in. We derive relational queries from real source text and gold evidence rather than report leaderboard results on the original benchmarks. Each domain contains 100 records allocated as 50 belonging, 25 before, and 25 after queries, for 200 records total. Each source question contributes at most one accepted record, and calibration and formal pools are separated.

Gold evidence is mapped to native target units before eligible anchors are enumerated. Belonging pairs share a container and are separated by at least four units; before and after pairs place the target four to eight units from the anchor in the specified direction. QASPER evidence must map unambiguously to paragraphs, while LongMemEval uses the answer-session identifier and turn-level \texttt{has\_answer} marker. Frozen hash ordering and rule-based validation enforce the relation quotas, coordinate constraints, evidence support, and leakage checks.

An instruction-following model assisted question wording at temperature zero, followed by agent review and local rule validation. This process is recorded as agent-assisted construction, not independent human annotation. Target-term refinement preserves phrases that can be located in the source. Some targets remain identifiable from $B$ alone, so the dataset tests whether relational scoping is useful, not whether topology is counterfactually necessary for every question. Model identity, prompts, review fields, and provenance are retained with the accompanying artifacts; Appendix~\ref{app:repro} summarizes data identity and run provenance.

Source versions, licenses, derived diagnostic material, and checksums appear in Appendix~\ref{app:repro}. Diagnostic records are not added to the 200-example main set.

\subsection{Methods and shared settings}

All methods use 768-dimensional normalized BGE-base-en-v1.5 embeddings and cosine similarity \cite{xiao2023cpack}, with a shared instruction prefix on the query side. Encoding length is explicitly capped at 512 tokens. Long native units are truncated only for automatic localization embeddings; original text used for coordinates and budget accounting is not truncated. The main evaluation forms nonoverlapping chunks from fixed 256 source-token spans (\texttt{fixed\_256}) with reversible character mappings and a common 1,024-token return limit. Because tokenization can differ across processing stages, a chunk re-encoded by BGE need not contain exactly 256 BGE tokens.

Table~\ref{tab:methods} distinguishes the information and role of each comparison. Dense RAG searches every chunk within a source instance. Anchor Window is an undirected $\pm8$ native-unit window ordered by distance to the anchor, with ties broken by native order and chunk identifier; it does not semantically rerank for the target. Because constructed targets are at least four units from the anchor and long units can exhaust the token budget, this is a generic distance-window baseline rather than a scope-matched implementation of the three relations.

\begin{table}[htbp]
\centering\small
\caption{Interfaces of the main comparisons. The last three methods receive the same correct anchor; all methods share chunks, encoder, and return budget.}
\label{tab:methods}
\begin{tabularx}{\linewidth}{@{}l>{\raggedright\arraybackslash}X>{\raggedright\arraybackslash}Xl@{}}
\toprule
Method & Candidate scope & Ranking input/basis & Role\\
\midrule
Dense RAG & Full instance & Full question $q$ & Flat dense baseline\\
Anchor Window & Anchor $\pm8$ units & Native distance & Distance-window control\\
\FQ{} & Chunks overlapping $\Gamma_r(a,T_X)$ & Full question $q$ & Internal ranking variant\\
\textbf{NativeScope} & Same as above & Target term $B$ & Target-ranking implementation\\
\bottomrule
\end{tabularx}
\end{table}

Dense RAG does not receive the gold anchor coordinate, so comparison with scoped methods measures the value of both structural input and its use. \FQ{} holds the candidate scope fixed and isolates ranking with $q$ versus $B$.

\subsection{Primary metric and statistics}

Let $G_i$ be the gold native-unit set for question $i$, and let $H_i$ be the union of character intervals from returned chunks. We define
\begin{equation}
R_i=\frac{1}{|G_i|}\sum_{u\in G_i}\mathbf{1}[I(u)\subseteq H_i],\qquad
\metric=\frac1N\sum_{i=1}^N R_i.
\end{equation}
The metric requires complete coverage of each gold native unit; returning only an answer span or part of a paragraph is insufficient. Results are macro-averaged over questions rather than weighted by the number of evidence units. A long dialogue turn may require multiple chunks for full coverage, making this metric stricter than answer-string matching. It is neither QA accuracy nor an official score from the source benchmarks.

For paired method differences on the same questions, we perform 10,000 paired bootstrap replications within each domain. Using NumPy's \texttt{default\_rng} with seed 20260929, each replicate independently samples with replacement 50 belonging, 25 before, and 25 after questions, recomputes the mean paired difference, and enters an ordinary 2.5/97.5 percentile interval. Relation-specific descriptions accompany aggregate results. Questions are the resampling units; intervals are not clustered by source paper or history, and no multiple-comparison correction is applied.

\subsection{Post-observation revision and result reuse}
\label{sec:revision}

The original frozen protocol used BGE Top-1 automatic anchoring. After observing localization errors and severe memory-domain degradation, we revised the main question to the conditional setting in which a correct anchor is supplied. This post-observation analysis uses all 200 records and leaves the questions, decompositions, evidence, encoder, chunking, and budget unchanged. The original automatic-localization results remain in Section~\ref{sec:auto} as an applicability boundary.

Existing outputs are reused only when their inputs and ranking rule match the revised condition; all other anchor-dependent results are recomputed. This decision depends on anchor correctness, not retrieval performance. Exact reuse counts, hashes, and method-specific provenance appear in Appendix~\ref{app:repro}. The revision isolates relational retrieval under a correct anchor but does not convert the analysis into a preregistered test; the earlier chunking, mechanism, and generative-QA runs remain outside the main evidence.

\section{Results}

\subsection{Main results with correct anchors}

Table~\ref{tab:main} reports results over all base records. NativeScope reaches 89.28\% in documents and 72.50\% in memories, improving over Dense RAG by 42.75 and 22.00 percentage points. Both paired intervals in Table~\ref{tab:ci} lie above zero, supporting the retrieval value of relational scoping on this controlled dataset under correct anchors and a fixed budget, rather than an unrestricted end-to-end RAG advantage.

\begin{table}[htbp]
\centering
\caption{$\metric$ with correct anchors (percent). Each domain has $N=100$ under fixed 256-token chunking. Bold marks only the largest point estimate.}
\label{tab:main}
\begin{tabular}{lrr}
\toprule
Method & Document & Memory\\
\midrule
Dense RAG & 46.53 & 50.50\\
Anchor Window & 45.36 & 2.00\\
\FQ{} (internal variant) & 87.78 & 68.50\\
\textbf{NativeScope} & \textbf{89.28} & \textbf{72.50}\\
\bottomrule
\end{tabular}
\end{table}

Figure~\ref{fig:main-results} places all four results on a common scale. Both relation-localized implementations substantially exceed instance-wide Dense RAG, while NativeScope and \FQ{} are close. The main evidence therefore points to candidate scoping rather than the necessity of $B$-only ranking.

\begin{figure}[tbp]
\centering
\includegraphics[width=0.94\linewidth]{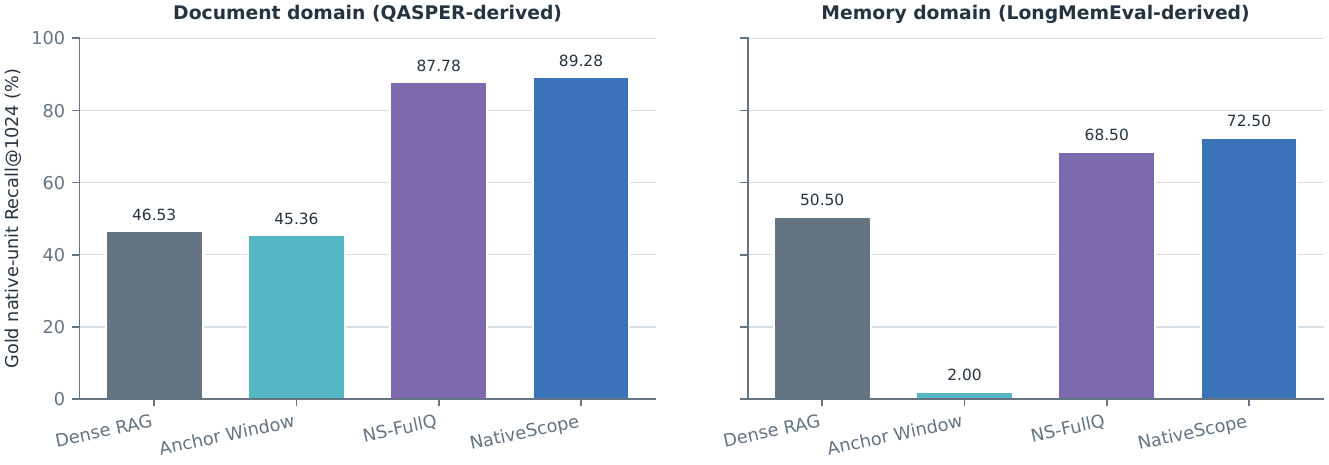}
\caption{Two-domain main results under correct anchors ($N=100$ per domain). The vertical axis is gold native-unit Recall@1024; all methods use the same chunks, encoder, and return budget.}
\label{fig:main-results}
\end{figure}

\begin{table}[htbp]
\centering\small
\caption{Paired NativeScope-minus-control differences and 95\% intervals, in percentage points. Bold marks comparisons whose interval lies entirely above zero; no advantage is claimed when an interval includes zero.}
\label{tab:ci}
\begin{tabular}{llrr}
\toprule
Domain & Control & Difference & 95\% interval\\
\midrule
Document & Dense RAG & $\boldsymbol{+42.75}$ & $\boldsymbol{[33.25,52.50]}$\\
Document & Anchor Window & $\boldsymbol{+43.92}$ & $\boldsymbol{[34.75,53.34]}$\\
Document & \FQ{} & $+1.50$ & $[-2.00,5.50]$\\
Memory & Dense RAG & $\boldsymbol{+22.00}$ & $\boldsymbol{[12.00,32.00]}$\\
Memory & Anchor Window & $\boldsymbol{+70.50}$ & $\boldsymbol{[61.50,79.00]}$\\
Memory & \FQ{} & $+4.00$ & $[-3.00,11.00]$\\
\bottomrule
\end{tabular}
\end{table}

The memory-domain window control reaches 2.00\%. As specified in the protocol, it prioritizes nearby units even though constructed targets begin four units away, and long turns can exhaust the budget before the target is covered. The comparison therefore shows that relation-aware scoping is preferable to this fixed distance policy under the present design.

\subsection{Relation types and the internal ranking variant}
\label{sec:ablation}

\FQ{} and NativeScope share anchors, native scopes, and candidate chunks on all 200 questions. Rankings differ on 171 records and returned chunk identifiers on 165, so the implementations are behaviorally distinct. The aggregate benefit of $B$ over $q$ is nevertheless only 1.50 and 4.00 points, and both intervals include zero. The present results support relational scoping but neither an independent advantage from $B$-only ranking nor statistical equivalence between the two variants.

\begin{table}[htbp]
\centering\small
\caption{$\metric$ by relation under correct anchors (percent). These groups partition the main evaluation. Bold marks the largest point estimate in each row; ties are included.}
\label{tab:relation}
\begin{tabular}{llrrrrr}
\toprule
Domain & Relation & $N$ & Dense & Window & \FQ{} & \textbf{NativeScope}\\
\midrule
Document & belonging & 50 & 46.06 & 53.56 & 86.06 & \textbf{92.06}\\
Document & before & 25 & 55.00 & 37.00 & \textbf{92.00} & 86.00\\
Document & after & 25 & 39.00 & 37.33 & \textbf{87.00} & \textbf{87.00}\\
Memory & belonging & 50 & 39.00 & 2.00 & 53.00 & \textbf{61.00}\\
Memory & before & 25 & 60.00 & 4.00 & 80.00 & \textbf{84.00}\\
Memory & after & 25 & 64.00 & 0.00 & \textbf{88.00} & 84.00\\
\bottomrule
\end{tabular}
\end{table}

NativeScope's point estimate exceeds Dense RAG in all six groups, but this does not establish a significant improvement in every group: the memory before and after intervals are $[0,48]$ and $[0,40]$ percentage points. The relative ranking input also varies by relation: $B$ is six points below the full question for document before and four points below for memory after. These results favor a relation-localized framework with a selectable ranking input rather than a requirement to discard the full question.

\subsection{Degradation under automatic localization}
\label{sec:auto}

Table~\ref{tab:auto} preserves the original automatic-localization results for default chunking. Top-1 anchor accuracy is only 61\% in documents and 37\% in memories; the proportions retaining all gold targets in the relation scope are 74\% and 53\%. Dense RAG does not consume the predicted anchor; memory-domain NativeScope reaches 35.50\%, below Dense RAG at 50.50\%, because an upstream coordinate error can become a hard candidate exclusion.

\begin{table}[htbp]
\centering\small
\caption{Default-chunking results under the original automatic-localization protocol (percent). Localization and scope coverage are diagnostics for the relation-localized pipeline. Within the retrieval-method rows, bold marks the largest recall in each domain.}
\label{tab:auto}
\begin{tabular}{lrr}
\toprule
Metric/method & Document & Memory\\
\midrule
Anchor Top-1 accuracy & 61.00 & 37.00\\
All gold targets in relation scope & 74.00 & 53.00\\
Dense RAG recall & 46.53 & \textbf{50.50}\\
Anchor Window recall & 43.69 & 3.00\\
\FQ{} recall & 65.44 & 33.50\\
\textbf{NativeScope recall} & \textbf{66.94} & 35.50\\
\bottomrule
\end{tabular}
\end{table}

Figure~\ref{fig:anchor-sensitivity} directly compares the two shared-scope implementations under correct and automatic anchors. The document domain drops by about 22 points, and the memory domain by more than 33, making anchor reliability a central condition for safe hard scoping.

\begin{figure}[tbp]
\centering
\includegraphics[width=0.90\linewidth]{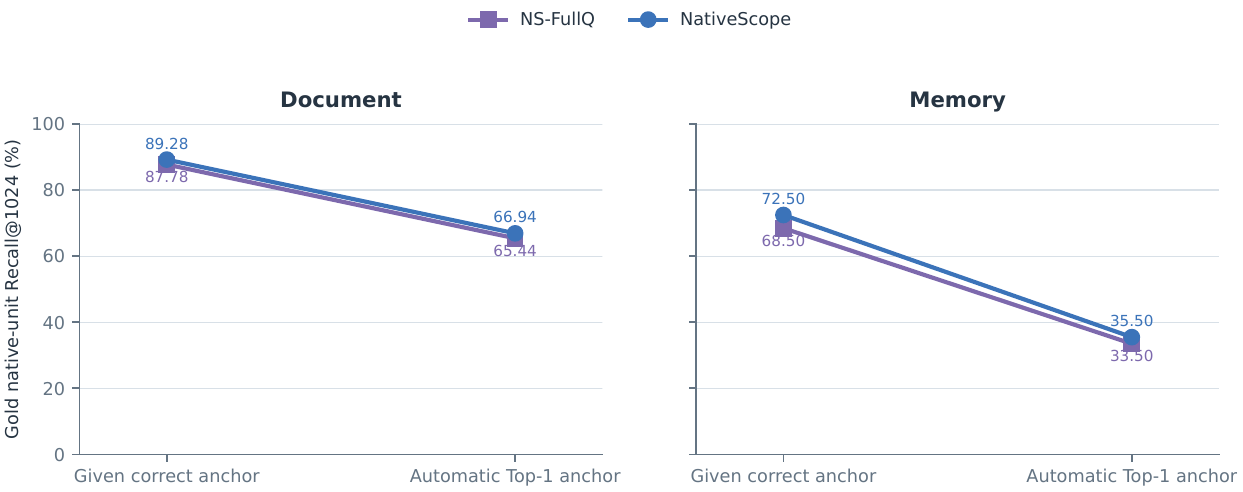}
\caption{Effect of the anchor condition on relation-localized recall. Correct-anchor and original automatic-localization results share default chunking and Recall@1024, while the main conditional evaluation was introduced after the automatic degradation was observed.}
\label{fig:anchor-sensitivity}
\end{figure}

On memory questions under automatic localization, NativeScope records 14 wins, 57 ties, and 29 losses relative to Dense RAG. Twenty-six of the losses have an incorrect anchor, and 25 exclude at least one gold target from the relation scope. One representative error maps a generic assistant opening to a similar turn in another session, thereby excluding the target before ranking begins. The pattern links most losses to the localization interface, although it does not assign a separate causal mechanism to every record.

\section{Discussion and Limitations}

The experiments support a simple division of labor: native relations determine where retrieval should search, and semantic similarity orders the candidates inside that scope. Under correct anchors, this interface substantially improves complete-evidence recall in both domains. The small and uncertain difference between NativeScope and \FQ{} locates the contribution in executable relational scoping, while leaving the best within-scope ranking query open.

Three limitations determine the scope of this evidence. First, the correct-anchor condition was introduced after automatic localization failed, making the main result conditional and post-observation. Second, each domain contains 100 agent-assisted controlled queries without independent human reannotation or fully natural external-query evaluation; literal phrase matching, minimum anchor distance, and complete-unit recall all shape the task. Third, the comparison set omits stronger lexical retrievers, rerankers, and external structured-retrieval systems. The study consequently establishes the utility of native relations in this controlled interface, not their necessity or global optimality for question answering.

The correct-anchor evaluation covers only the default chunking. Earlier chunking and mechanism runs use a different anchor condition and are retained as exploratory records rather than combined with the main analysis. Generative QA stopped after 36 of 800 planned responses, and latency, memory, and API cost were not evaluated. These omissions leave cross-chunking robustness, end-to-end answer quality, and systems efficiency for future work.

In deployment, reliable coordinates may come from user selection, stable record identifiers, or a validated upstream locator. When only unreliable automatic localization is available, calibration, verification, or fallback is necessary because hard scoping can underperform instance-wide retrieval. Errors in container metadata, cross-session events, and ambiguous relation parsing remain outside the current interface guarantees.

\section{Conclusion}

NativeScope executes stored relations to restrict the candidate domain before semantic ranking. On 200 controlled records with correct anchors, it improves native-unit recall over instance-wide Dense RAG by 42.75 points for documents and 22.00 points for memories. Full-question ranking performs similarly within the same scope, while automatic-anchor failures sharply reduce recall. The central result is therefore conditional but actionable: when a system can supply reliable coordinates and native relations, executing that structure can materially improve retrieval under a fixed context budget.

\appendix
\section{Reproduction Information and Evidence Inventory}
\label{app:repro}

\subsection{Versions and data identity}

Raw datasets are not redistributed with the derived package, and API keys do not enter the paper or artifacts. QASPER is recorded under CC BY 4.0, LongMemEval cleaned under MIT, and BGE weights under MIT. Full raw-file SHA-256 values, derived-set hashes, and model-file inventories appear in \path{experiments/E/data_manifest.json}. Key identities are:
\begin{quote}\small
QASPER revision: \path{fdc9d8214fbab5dd782958601db4d678e6934a54}; v0.3 test.\\
LongMemEval revision: \path{98d7416c24c778c2fee6e6f3006e7a073259d48f}; S cleaned only.\\
BGE revision: \path{a5beb1e3e68b9ab74eb54cfd186867f64f240e1a}.\\
Base 200-record SHA-256: \texttt{5de99299a64db82e\allowbreak f3b8aa5277745e93\allowbreak dc1693185d06b226\allowbreak e452e84a8162b678}.
\end{quote}
Native parsing retains paragraph or turn text, container identifiers, and global unit order. Frozen records contain source-question identifiers, source revisions, candidate hashes, prompt versions, and reviewer types. Derived data and audit logs are not equivalent to human gold annotation; reusers should retain this provenance statement.

\subsection{Run scope and result provenance}

The original automatic-localization experiment contains 40 configurations and 4,000 question--method results spanning two domains, four methods, and five chunkings. The mechanism batch contains 102 configurations and 9,280 results, including ten-seed topology shuffling, full-question gold-anchor diagnostics, three semantic conditions, a positive lexical-matching subset, and role/reverse checks. Except for explicitly labeled gold-anchor diagnostics, these runs use automatic localization. Table~\ref{tab:inventory} states how they contribute to the paper.

\begin{table}[htbp]
\centering\small
\caption{Executed material and its evidential role. Completed runs do not establish all earlier claims.}
\label{tab:inventory}
\begin{tabularx}{\linewidth}{@{}lrl>{\raggedright\arraybackslash}X@{}}
\toprule
Batch & Results & Anchor condition & Use in this paper\\
\midrule
Original chunking main run & 4,000 & Automatic & Default-chunking boundary result; other runs retained as exploratory records\\
Original mechanism batch & 9,280 & Automatic and separate diagnostics & Audit retained; not evidence for mechanisms under the revised condition\\
Correct-anchor revision & 800 & Given correct coordinate & Main evidence in Tables~\ref{tab:main}--\ref{tab:relation}\\
Generative QA & 36 & Original protocol, incomplete & Paused; no method comparison reported\\
\bottomrule
\end{tabularx}
\end{table}

The revision note is \path{experiments/E/GIVEN_ANCHOR_REVISION.md}. The independent result directory \path{experiments/E/results/given_anchor_v1/} contains per-record results, reuse provenance, run status, and statistics. \path{analysis.json} stores unrounded means, paired differences, and intervals; table percentages and percentage points are these values multiplied by 100 and rounded. Across all four methods, the document domain's 400 results comprise 322 reused and 78 recomputed, while the memory domain's 400 comprise 274 reused and 126 recomputed; these reused counts include the 100 Dense RAG results in each domain. The anchor intervention replaces only the localization return value. NativeScope continues to rank with $B$ and uses unchanged scope and budget code.

\subsection{Verification boundary}

Data validation checks coordinates, containers, directional distances, evidence mappings, uniqueness, duplicates, and quotas. Retrieval records preserve the native scope, candidate chunks, ranking, and returned chunks so coverage and budget can be recomputed. Consistency checks between full-question and $B$ ranking must compare candidate sets rather than infer identical behavior from similar recall. The original gold-anchor diagnostic ranks with the full question and must be interpreted as \FQ{} when reproduced.

\bibliographystyle{unsrt}
\bibliography{E_nativescope_en}
\end{document}